\documentclass[conference]{IEEEtran}
\ifCLASSINFOpdf
\else
\fi
\usepackage[cmex10]{amsmath}
\usepackage{amsthm}
\usepackage{bm}
\usepackage[final]{graphicx}
\usepackage{amsfonts}

\usepackage{subfigure}
\usepackage{caption}

\newtheorem{lemma}{Lemma}
\newtheorem{theorem}{Theorem}

\newtheorem{cor}{Corollary}

\begin{document}
%
\title{Reverse Edge Cut-Set Bounds for Secure Network Coding}

\author{\IEEEauthorblockN{Wentao Huang and Tracey Ho}
\IEEEauthorblockA{Department of Electrical Engineering\\
California Institute of Technology\\
\{whuang,tho\}@caltech.edu}
\and
\IEEEauthorblockN{Michael Langberg}
\IEEEauthorblockA{Department of Electrical Engineering\\
University at Buffalo, SUNY\\
mikel@buffalo.edu}
\and
\IEEEauthorblockN{Joerg Kliewer}
\IEEEauthorblockA{Department of ECE\\
New Jersey Institute of Technology\\
jkliewer@njit.edu}}


%


\maketitle

\begin{abstract}
We consider the problem of secure
communication over a network in the presence of wiretappers.
We give a new cut-set bound on secrecy capacity which takes into
account the contribution of both forward and backward edges
crossing the cut, and the connectivity between their endpoints
in the rest of the network. We show the bound is tight on a
class of networks, which demonstrates that it is not possible to
find a tighter bound by considering only cut set edges and their
connectivity.
\end{abstract}


%
\IEEEpeerreviewmaketitle

\allowdisplaybreaks[4]
\section{Introduction}
Consider a noise free communication network in which an information source $S$ wants to transmit a secret message to the destination $D$ over the network in the presence of a wiretapper who can eavesdrop a subset of edges. The secure network coding problem, introduced by Cai and Yeung \cite{Cai02}, studies the secrecy capacity of such networks. Under the assumptions that 1) all edges have unit capacity; 2) the wiretapper can eavesdrop any subset of edges of size up to $z$; 3) only $S$ has the ability to generate randomness, \cite{Cai02} shows that the secrecy capacity is $x-z$, where $x$ is the min-cut from $S$ to $D$. Subsequent works have studied various ways to achieve this capacity with codes on fields of smaller size \cite{Feldman04}, coset codes \cite{Rouayheb12}, and universal codes \cite{Silva10}.

Though the secrecy capacity is well understood in this special case, much less is known under a more general setting. In particular, if either edge capacities are not uniform, or the collection of possible wiretap sets is more general (i.e., not characterized by a simple parameter $z$), Cui et al. \cite{Cui13} show that finding the secrecy capacity is NP-hard. On the other hand, if randomness is allowed to be generated at non-source nodes, Cai and Yeung ~\cite{CY07} give an example in which this can be advantageous, and provide a necessary and sufficient condition for a linear network code to be secure. However, for this case \cite{Huang13, chan_grant_08} show that finding the secrecy capacity is at least as difficult as the long-standing open problem of determining the capacity region of multiple-unicast network coding. To the best of our knowledge, under these general settings, the only known bounds of secrecy capacity are given implicitly in terms of entropy functions/entropic region \cite{Chan08, jalali12capacity}, whereas determining the entropic region is a long standing open problem as well.

This paper gives the first explicit upper bound on secrecy capacity for the secure network coding problem in the case where non-source nodes can generate randomness.
Our bound is based on cut-sets and has an intuitive graph-theoretic interpretation. The key observation is that unlike traditional cut-set bounds which only consider forward edges, for the secure network coding problem backward edges may also be helpful in a cut if down-stream (hence non-source) nodes can generate randomness, as shown in Fig. \ref{bwdhelp}-(a). Here the backward edge $(A,S)$ can transmit a random key back to the source to protect the message, and enable secrecy rate 1 to be achieved. 
However, one should be careful in counting the contribution of backward edges since they are not always useful, such as edge $(D,A)$ in Fig. \ref{bwdhelp}-(b). Notice that the networks of (a) and (b) are identical from the perspective of cuts because they each contain a cut with two forward edges and a cut with one forward edge and one backward edge. Hence to avoid a loose bound we have to see beyond the cut: in this simple example the backward edge in (a) is helpful because it is connected to the forward edge, while the one in (b) is not. More generally, this motivates us to take into account the connectivity from backward edges to forward edges, described by a 0-1 connectivity matrix $C$. We show that the rank structure of the submatrices of $C$ characterizes the utility of the backward edges, and use this to obtain an upper bound on secure capacity.

\begin{figure}[htpb]
\centering
  \subfigure[Backward edge helpful]{\includegraphics[width=0.21\textwidth]{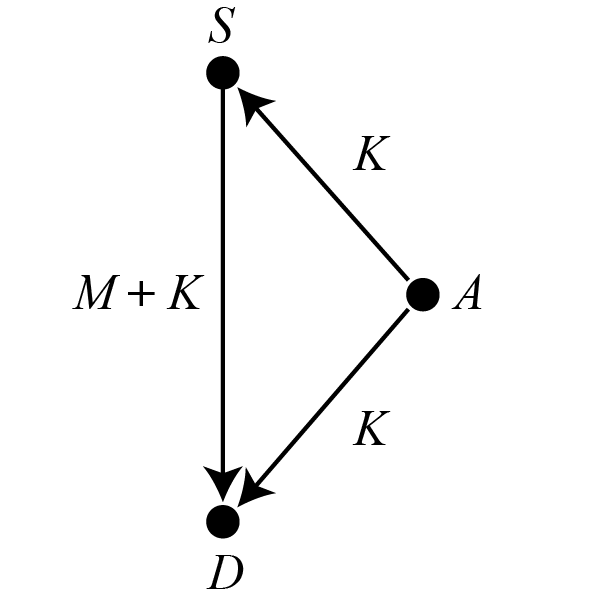}}
  \subfigure[Not helpful]{\includegraphics[width=0.21\textwidth]{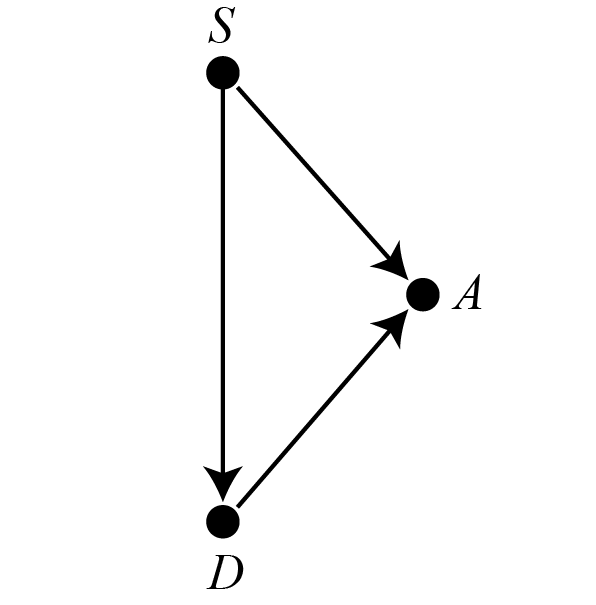}}
  \caption{Networks with unit capacity edges and $z=1$.}\label{bwdhelp}
\end{figure}

Finally we show that given any network cut, we can construct a network with the same cut set edges and connectivity between their endpoints, such that our bound is achievable by random scalar linear codes. Hence the bound is optimal in the sense
that it is not possible to find a better bound by merely considering the cut set edges and their connectivity.

\section{Models}
Consider a directed network $\mathcal{G}=(\mathcal{V},\mathcal{E})$ and let $\mathcal{A} \subset 2^{\mathcal{E}}$ be a collection of  wiretap sets. Since $\mathcal{A}$ is arbitrary (i.e., non-unform), without loss of generality we may assume all edges have unit capacity, because any edge of larger capacity can be replaced by a number of parallel unit capacity edges in both $\mathcal{G}$ and $\mathcal{A}$. In this work we focus on the single source single terminal setting. This seemingly simple setting is as at least as hard as determining the capacity region of multiple unicast network coding \cite{Huang13}. Let $S$ be the source and $D$ be the sink, $S$ wants to deliver a secret message $M$ to $D$ under perfect secrecy with respect to $\mathcal{A}$, i.e., denote $X(A)$ as the signals transmitted on $A \subset \mathcal{E}$, then $\forall A \in \mathcal{A}$, $I(M;X(A))=0$.
For all $i \in \mathcal{V}$, denote by $K_i$ the independent randomness generated at node $i$ that might be used as keys to protect the message.

Consider an arbitrary cut $V \subset {\mathcal{V}}$ such that $S \in V$ and $D \in V^c$. Denote $E_V^{\text{fwd}} = \{(i,j) \in \mathcal{E}: i \in V, j \in V^c\}$ as the set of forward edges with respect to $V$, and $E_V^{\text{bwd}} = \{(i,j) \in \mathcal{E}: i \in V^c, j \in V\}$ as the set of backward edges. Assume $|E_V^\text{fwd}|=x$ and $|E_V^\text{bwd}|=y$, we denote the $x$ forward edges by $e^\text{fwd}_1, e^\text{fwd}_2, ..., e^\text{fwd}_x$, and the $y$ backward edges by $e^\text{bwd}_1, e^\text{bwd}_2, ..., e^\text{bwd}_y$. Let $C_{b \to f} = ( c'_{ij} )$ be an $x \times y$ (0-1) matrix characterizing the connectivity from the backward edges to the forward edges. More precisely,
\begin{align*}
  c'_{ij} = \left\{ \begin{array}{ll}
    1 & \begin{array}{l}
      \text{if }\exists \text{ a directed path from head}(e^\text{bwd}_j) \text{ to tail}(e^\text{fwd}_i) \\ \text{that does not pass through any nodes in $V^c$}
    \end{array}
     \\
    0 & \begin{array}{l}\text{otherwise}    \end{array}
  \end{array} \right.
\end{align*}

\section{Cut-set Bound}
This section gives a cut-set bound of the secure capacity with respect to the cut $V$ and its connectivity matrix $C_{b \to f}$. We first prove a lemma before formally introducing the bound.
\begin{lemma}
 Given an arbitrary (0-1) matrix $C=(c_{ij})$ of size $a \times b$ and $\mathcal{U}$ a collection of submatrices of $C$, there is a large enough $q$ such that there exists a matrix $\bar{C} \in \mathbb{F}_q^{a\times b} = (\bar{c}_{ij})$ with following properties: 1) $\bar{c}_{ij}=0$ if $c_{ij}=0$; 2) $\forall$ $U \in \mathcal{U}$, assume its size is $m \times n$ and let the corresponding submatrix of $\bar{C}$ be $\bar{U}$, then rank$(\bar{U}) = \max_{V \in \mathbb{F}_q^{m \times n}, v_{ij}=0 \text{ if } u_{ij}=0}$ rank$(V)$, i.e., $\bar{U}$ is rank maximized subject to the zero constraints given in $C$. In particular, $q > |\mathcal{U}|ab$ is sufficient.
\end{lemma}
\begin{proof}
  Consider a finite field $\mathbb{F}_q$ of order $q$, and any $U \in \mathcal{U}$, let $$\bar{V} = \arg \max_{V \in \mathbb{F}_q^{m \times n}, v_{ij}=0 \text{ if } u_{ij}=0} \text{rank}(V),$$ and let $r_U = \text{rank}(\bar{V})$. So $\bar{V}$ contains an $r_U \times r_U$ full rank submatrix, denoted by $\underline{V}$. 
  Let $\underline{C}=(\underline{c}_{ij})$ be the submatrix of $C$ corresponding to the position of $\underline{V}$. Now consider a polynomial matrix $\underline{V}[\bm{x}] = (\underline{v}_{ij})$ defined by
  \begin{align*}
    \underline{v}_{ij} = \left\{ \begin{array}{ll}
      0 & \text{if } {\underline{c}}_{ij}=0\\
      x_{ij} & \text{if } {\underline{c}}_{ij}=1
    \end{array}\right.
  \end{align*}
  where the $x_{ij}$'s are indeterminates. Then it follows that $\det( \underline{V}[\bm{x}] )$ is not the zero polynomial because otherwise $\det(\underline{V}) = 0$ and $\underline{V}$ cannot be full rank. Now let the non-zero entries of $\underline{V}[\bm{x}]$, i.e., all the $x_{ij}$'s, be i.i.d. uniformly distributed on $\mathbb{F}_q$. By the Schwartz-Zippel lemma,
  \begin{align*}
    \Pr\left\{ \det(\underline{V}[\bm{x}]) = 0 \right\} \le \frac{r_U^2}{q} \le \frac{ab}{q}
  \end{align*}
  Notice that the polynomial matrix $\underline{V}[\bm{x}]$ is in fact a submatrix of a $a \times b$ polynomial matrix $B[\bm{x}]=(b_{ij})$ defined by
   \begin{align*}
    b_{ij} = \left\{ \begin{array}{ll}
      0 & \text{if }c_{ij}=0\\
      x_{ij} & \text{if }c_{ij}=1
    \end{array}\right.
  \end{align*}
  where again the $x_{ij}$'s are indeterminates. Let all the non-zero entries of $B[\bm{x}]$ follow i.i.d. uniform distribution on $\mathbb{F}_q$, and by the union bound, we have
  \begin{align*}
    \Pr\left\{ \bigcup_{U \in \mathcal{U}} \det(\underline{V}[\bm{x}]) \ne 0 \right\} & \ge 1 - \sum_{U \in \mathcal{U}} \Pr\{ \det(\underline{V}[\bm{x}]) = 0 \}\\
    & \ge 1 - |\mathcal{U}| \frac{ab}{q}
  \end{align*}
  Therefore if $q > |\mathcal{U}|ab$, there exists an evaluation of $B[\bm{x}]$ such that $\det(\underline{V}[\bm{x}]) \ne 0$ for any $U \in \mathcal{U}$. This evaluation gives a desired $\bar{C}$, because for any $U \in \mathcal{U}$, the corresponding submatrix $\bar{U}$ of $\bar{C}$ contains a full rank square submatrix of size $r_U$, and by definition $r_U$ is the maximum rank $\bar{U}$ can achieve subject to the zero constraints in $C$.
\end{proof}

Define
\begin{align*}
    C = \left( \begin{array}{l}
      C_{b\to f}\\I_y
    \end{array}
     \right),
\end{align*}
where $I_y$ is the identity matrix of order $y$. Notice that rows in $C$ correspond to edges crossing the cut in $\mathcal{G}$. Denote $\mathcal{A}_V = \{A \cap ( E_{V}^{\text{fwd}} \cup E_{V}^{\text{bwd}} ) : A \in \mathcal{A}\}$. For $A \in \mathcal{A}_V$, denote $U_A$ the submatrix of $C$ formed by the rows corresponding to edges in $A$. Let $\mathcal{U} = \{U_A, A \in \mathcal{A}_V \}$, and
 let $\bar{C}$ be the rank maximized matrix specified in Lemma 1 with respect to $C$ and $\mathcal{U}$. For $U_A \in \mathcal{U}$, let $\bar{U}_A$ be the corresponding submatrix of $\bar{C}$. We are now ready to state our main result.
\begin{theorem}\label{th1}
 The secrecy capacity is bounded by
\begin{align*}
  \mathfrak{C} \le x + \min_{A \in \mathcal{A}_V} \text{\emph{rank}}(\bar{U}_A) - |A|
\end{align*}
\end{theorem}

In the special case of uniform wiretap sets, i.e., $\mathcal{A} = \{ A \subset \mathcal{E} : |A| \le z \}$, Theorem \ref{th1} reduces to the following form.
\begin{cor}
  Define $k_b = \min_{\{\bar{U}:\text{\emph{ $z \times y$ submatrix of }} \bar{C}\}} \text{\emph{rank}}(\bar{U}),$ then the secrecy capacity is bounded by
  \begin{align*}
  \mathfrak{C} \le x + k_b -z
\end{align*}
\end{cor}

In what follows, we will prove Theorem 1. Given a cut of $x$ forward edges, $y$ backward edges, and the connectivity matrix $C_{b\to f}$, we construct an upper bounding network $\bar{\mathcal{G}}$ as follows: 1) Absorb all nodes downstream the cut, i.e., all $v \in V^c$, into the sink $D$. So for all $i,j$, head($e^\text{fwd}_i$)=$D$, tail($e^\text{bwd}_j$)=$D$. 2) Connect the source to each forward edge with infinite unit capacity edges $(S, \text{tail}(e^\text{fwd}_i))$.
3) Connect the backward edges to the forward edges according to $C_{b \to f}$. More precisely, add an infinite amount of unit capacity edges $(\text{head}(e^\text{bwd}_j),\text{tail}(e^\text{fwd}_i))$ if and only if $c'_{ij}=1$. Finally, in $\bar{\mathcal{G}}$ we only allow $S$ and $D$ to generate independent randomness.

\begin{lemma}
  The secure unicast capacity of $\bar{\mathcal{G}}$ upperbounds the secure unicast capacity of $\mathcal{G}$.
\end{lemma}
\begin{proof}[Proof Sketch]
  Note that all infinite parallel unit capacity edges are perfectly secure because they can be protected by an infinite number of local keys. Hence for any coding scheme on $\mathcal{G}$, the same coding scheme can be simulated on $\bar{\mathcal{G}}$ securely.  The assumption that only $S$ and $D$ can generate randomness is optimal, because if any other node wishes to generate independent randomness, such randomness may be generated at $S$ and sent to the node through the infinite parallel edges.
\end{proof}

Due to the fact that $\bar{\mathcal{G}}$ has a simplified structure, in the proof of Theorem 1 we shall always consider $\bar{\mathcal{G}}$ instead of $\mathcal{G}$ unless otherwise specified. Note that in $\bar{\mathcal{G}}$ only the edges crossing the cut $V$ are vulnerable, hence we may assume any wiretap set only contains these edges. Therefore $\mathcal{A}=\mathcal{A}_V$ and for notational convenience in what follows we no longer distinguish them. Let $F_1, ..., F_x$ be the signals transmitted on edges $e^\text{fwd}_1, ..., e^\text{fwd}_x$; $B_1, ..., B_y$ be the signals transmitted on edges $e^\text{bwd}_1, ..., e^\text{bwd}_y$. 
 Consider any $A \in \mathcal{A}$ and the set of signals $f_A = \{ \cup f_e: e \in A\}$ defined as follows.
\begin{align*}
  f_e = \left\{
  \begin{array}{ll}
\{B_j\} & \text{ if } e=e^\text{bwd}_j\\
\{B_j : c'_{ij}=1\} & \text{ if } e=e^\text{fwd}_i
  \end{array}
   \right.
\end{align*}
The following lemma shows that the rank structure of the submatrices of $\bar{C}$ has interesting properties.

\begin{lemma}\label{lemimport}
For any $A \in \mathcal{A}$, there exists a partition $A = A_1 \cup A_2$, such that $|f_{A_1}| + |A_2| = $ rank($\bar{U}_A$).
\end{lemma}
\begin{proof}
  The idea of the proof is to infer the structure of $U_A$ given the rank of $\bar{U}_A$ and the fact that $\bar{U}_A$ is rank maximized. Then since $U_A$ characterizes the connectivity to the edges in $A$ it becomes convenient to bound the size of $f_A$.

  Denote for short $r=\text{rank}(\bar{U}_A)$, so $\bar{U}_A$ contain an $r \times r$ submatrix whose determinant is non-zero, and therefore in $\bar{U}_A$ there exist $r$ non-zero entries at different columns and at different rows. Recall that an entry in $\bar{U}_A$ can be  non-zero only if this entry is 1 in $U_A$, hence $U_A$ contains $r$ entries of value 1 at different columns and different rows. Perform column and row permutations to move these 1's such that $U_A(r+1-i,i)=1, \forall 1\le i \le r$, i.e., they become the counter-diagonal entries of the upper-left block formed by the first $r \times r$ entries. See Fig. \ref{fig:eg} for an example. Note that permutations in $U_A$ are merely reordering of edges, and for notational convenience we denote the matrix after permutations as $U_A$ still.

  It then follows that $U_A(i,j)=0, \forall r<i\le |A|, r<j \le y$. Otherwise if any entry in this lower right block is non-zero, setting this and the aforementioned $r$ counter-diagonal entries as 1, and all other entries as 0 yields a matrix that satisfies the zero constraint in $U_A$ and it has rank $r+1$. But this is a contradiction because $r=\text{rank}(\bar{U}_A)$ is the maximum rank. Hence we label this block as $\emph{zero}$.

  Below we introduce an algorithm that further permutes $U_A$ and labels it blockwise. The algorithm takes a matrix $G$ of arbitrary size $m \times n$ and a positive integer parameter $k$ as input, such that the upper-left block $G_{UL}$ formed by the first $k \times k$ entries of $G$ has all 1's in its counter-diagonal. Now consider $G_{LL}= (g_{ij}), k < i \le m, 1 \le j \le k$ which is the lower-left block of $G$. If every column of $G_{LL}$ is non-zero, label this block as \emph{non-zero}, label $G_{UL}$ as \emph{counter-diagonal}, label the block $G_{UR}=(g_{ij}), 1 \le i \le k, k < j \le n $ as \emph{zero}*, return $t:=0$ and terminate. If $G_{LL}=\bm{0}$ or $G_{LL}$ is empty, label this block (if not empty) as \emph{zero}, label $G_{UL}$ as \emph{counter-diagonal}, label $G_{UR}$ as \emph{arbitrary}, return $t:=k$ and terminate.

  Otherwise $G_{LL}$ contains both zero and non-zero columns. In this case, first perform column permutations in $G$ to move all non-zero columns of $G_{LL}$ to the left and zero columns to the right. Assume that after permutation the first $u$ columns of $G_{LL}$ are non-zero, and the last $v$ columns are all zero. Label the block $(g_{ij}), k < i \le m, 1 \le j \le u$ as \emph{non-zero} and label the block $(g_{ij}), k < i \le m, u < j \le k$ as \emph{zero}. At this point some of the 1's originally in the counter-diagonal of $G_{UL}$ are misplaced due to column permutations, perform row permutations to move them back to the counter-diagonal. Note that only the first $k$ rows need to be permuted and the lower labeled block(s) is not affected. Label the block $(g_{ij}), k-u+1 \le i \le k , 1 \le j \le u$ as \emph{counter-diagonal}, label the block $(g_{ij}), 1 \le i \le k-u , 1 \le j \le u$ as \emph{arbitrary}, and label the block $(g_{ij}), k-u+1 \le i \le k , k < j \le n$ as \emph{zero}*. Then truncate the first $u$ columns and the last $m-k$ rows from $G$. Notice that the block formed by the first $v \times v$ entries in the truncated $G$ has all 1's in its counter-diagonal. Now invoke the algorithm recursively to the truncated $G$ with parameter $v < k$. The algorithm must terminate because the input parameter is a positive finite integer and cannot decrease indefinitely.

  Applying the algorithm to the matrix $U_A$ with parameter $k:=r$ will permute the rows and columns of $U_A$ and label it completely. Refer to Figure \ref{fig:eg} for an example. Notice that the algorithm always labels \emph{counter-diagonal}, \emph{non-zero} and \emph{zero} literally, i.e.,  by hypothesis all counter-diagonal-label blocks are square and have 1's in their counter-diagonals (but the off-counter-diagonal entries may be arbitrary); all non-zero-label blocks do not contain zero columns; and all zero-label blocks are all zero.   The only non-trivial label is \emph{zero}*, and we claim that the algorithm also labels \emph{zero}* correctly in the sense that a zero*-labeled block is indeed zero.

  \begin{figure}[htpb]
  \begin{center}
      \includegraphics[width=0.43\textwidth]{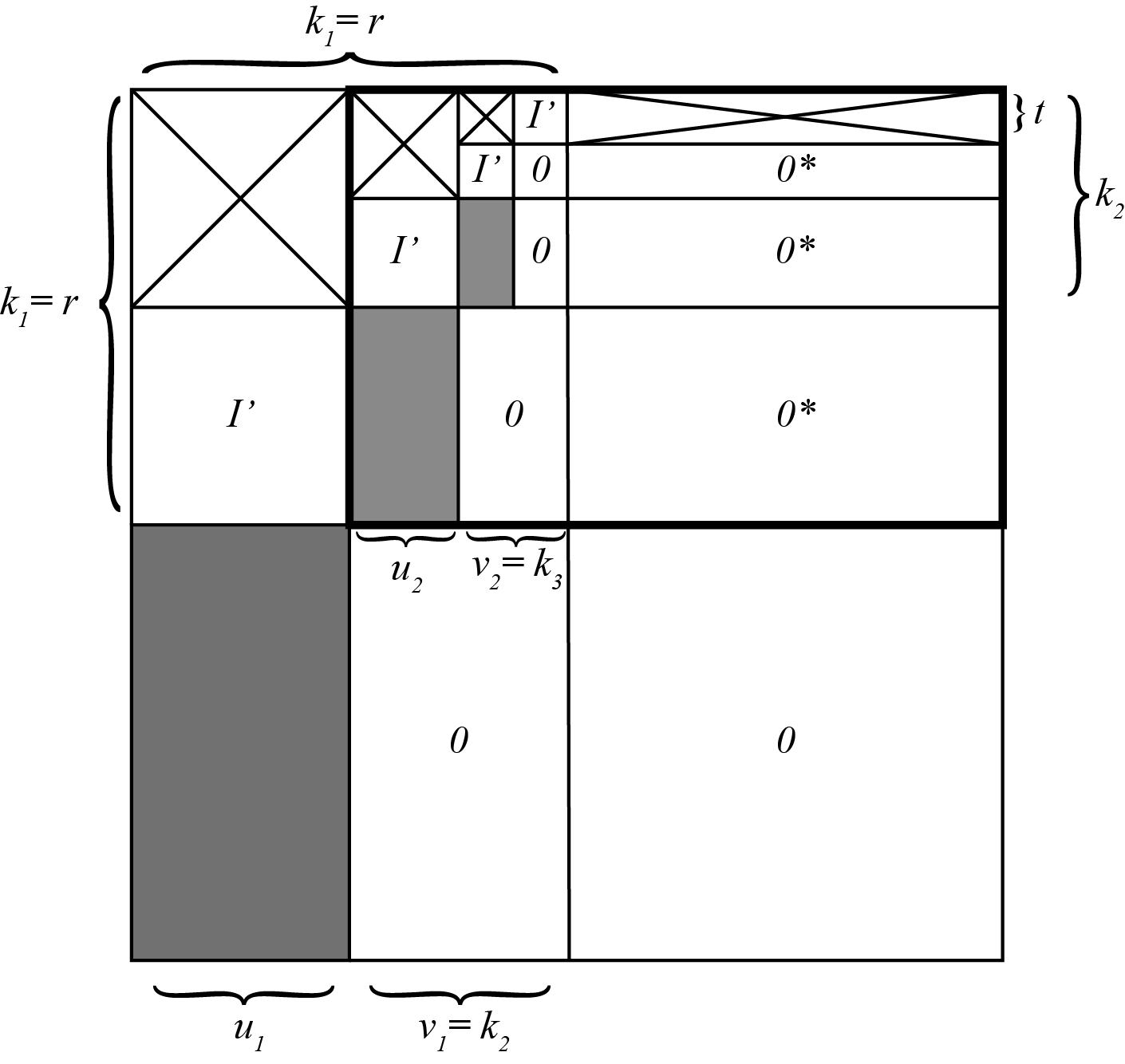}
  \caption{An example of a labeled $U_A$. \emph{zero} blocks are indicated by $0$; \emph{zero}* blocks are indicated by $0$*; \emph{counter-diagonal} blocks are indicated by $I'$; \emph{non-zero} blocks are colored in gray and \emph{arbitrary} blocks are crossed. The algorithm terminates in four iterations and returns $t$. Key parameters of the first two iterations are illustrated and the subscripts denote iteration numbers. The truncated $G$ after the first iteration is highlighted in bold line. Note that the first $t$ rows correspond to $A_2$, and the remaining rows correspond to $A_1$.}\label{fig:eg}
  \end{center}
\end{figure}

  To prove the claim, notice that all \emph{zero}* blocks pile up at the last $y-r$ columns of $U_A$, and consider any entry $\alpha_1$ of a \emph{zero}* block. By the algorithm the row of $\alpha_1$ must intersect a unique \emph{counter-diagonal} block, and denote the intersecting counter-diagonal entry of the \emph{counter-diagonal} block as $\beta_1$.
   By the algorithm this intersecting \emph{counter-diagonal} block must lie immediately on top of a \emph{non-zero} block. Therefore the lower \emph{non-zero} block contains a non-zero entry $\alpha_2$ in the same column as $\beta_1$. And again the row of $\alpha_2$ will intersect a counter-diagonal entry $\beta_2$ of a \emph{counter-diagonal} block. In exactly the same way we are able to find a sequence of entries $\alpha_3, \beta_3, \alpha_4, \beta_4 ...$ until we reach the lowest \emph{non-zero} block. Note that all these entries belong to distinct blocks, and because there is a finite number of blocks, the series is finite. In particular, let $w$ be the number of \emph{counter-diagonal} blocks that lie below or intersect the row of $\alpha_1$, then we can find $\beta_1, ..., \beta_w$ and $\alpha_1, ..., \alpha_{w+1}$, where $\alpha_{w+1}$ lies in the lowest \emph{non-zero} block. Now suppose for the sake of contradiction that $\alpha_1$ is non-zero, set $\alpha_1, ..., \alpha_{w+1}$ to 1, set all counter-diagonal entries of all \emph{counter-diagonal} blocks except $\beta_1, ... , \beta_w$ to 1, and set all other entries to 0. This produces a matrix of rank $r+1$ because all $r+1$ 1's appears in distinct columns and rows, which contradicts the fact that $\bar{U}_A$ is rank maximized.

  Hence all zero*-label blocks are indeed zero. In particular, after the permutations, the block $U_A(i,j), t < i \le |A| , r-t +1 \le j \le y$ is all zero.
  Now partition $A$ into $A_1 \cup A_2$, where $A_2$ is the subset of edges corresponding to the first $t$ rows of the permuted $U_A$. So $|A_2|=t$ and $|A_1|=|A|-t$. But the zero constraints in $U_A$ imply that $f_{A_1}$ contains $r-t$ of the $B_j$'s corresponding to the first $r-t$ columns, hence $|f_{A_1}| = r-t$. Finally $ |f_{A_1}| + |A_2| = r$.
\end{proof}

\begin{cor}\label{cor1}
  Partition $A$ into $A_1 \cup A_2$ as in Lemma \ref{lemimport}. Further partition $A_1$ as $A_F \cup A_B$, where $A_F \subset \{e^\text{fwd}_1, ..., e^\text{fwd}_x\}$, $A_B \subset \{e^\text{bwd}_1, ..., e^\text{bwd}_y\}$, then $H(f_{A_F}| f_{A_B}) \le \text{rank}(\bar{U}_A) - |A_B| - |A_2|$.
\end{cor}
\begin{proof}
  Suppose for contradiction that $H(f_{A_F}| f_{A_B}) > \text{rank}(\bar{U}_A) - |A_B| - |A_2|$, then $|f_{A_F} \backslash f_{A_B}| \ge H(f_{A_F} \backslash f_{A_B}) \ge H(f_{A_F}|f_{A_B}) > \text{rank}(\bar{U}_A) - |A_B| - |A_2|$. This implies $|f_{A_{1}}| > \text{rank}(\bar{U}_A) - |A_2|$, a contradiction to Lemma \ref{lemimport}.
\end{proof}

Due to the cyclic nature of $\mathcal{G}'$, imposing delay constraints on some edges is necessary to avoid stability and causality issues.  It suffices to assume there is unit delay on edges $e_1^\text{fwd},..., e_x^\text{fwd}, e_1^\text{bwd} , ..., e_y^\text{bwd}$. Note that any realistic systems should comply with these minimal delay constraints, e.g., it is not possible that a forward signal $F_i$ is a causal output depending on a backward signal $B_j$, while $B_j$ is also a causal output depending on $F_i$. Let $t$ be a time index, denote $F_i[t]$ and $B_j[t]$ as the signals transmitted on edges $e^\text{fwd}_i$ and $e^\text{bwd}_j$ during the $t$-th time step. Consider an arbitrary secure coding scheme that finishes within $T$ time steps. Below we show that the rate of this code is upper bounded by $x+ \text{rank}(\bar{U}_A) - |A|$, $\forall A \in \mathcal{A}$, as claimed in Theorem \ref{th1}. We first prove a lemma.


\begin{lemma}\label{takeoutD}
  Consider arbitrary random variables $X,Y,Z,W$, if $(Z, W) \to (Y, W) \to X$, then
  \begin{align*}
    H(X | Z , W ) \ge H(X | W) - I(Y;X | W)
  \end{align*}
\end{lemma}
\begin{proof}
Note that $H(X, Y|W) = H(X|Y, W) + H(Y| W) = H(Y|X, W) + H(X| W)$.
So $H(X|Y, W) = H(X|W) + H(Y|X, W) - H(Y | W) = H(X|W) - I(Y;X|W)$. Finally because $H(X|Z, W) \ge  H(X|Y,W)$, we prove the claim.
\end{proof}

\begin{proof}[Proof (of Theorem \ref{th1})]
Define $\mathcal{F}[t] = \{ F_1[t], ..., F_x[t]\}$ as all the forward signals at time $t$, and $\mathcal{B}[t] = \{ B_1[t], ..., B_y[t]\}$ as all the backward signals. Let $\mathcal{F} = \{\mathcal{F}[1], ..., \mathcal{F}[T]\}$, $\mathcal{B} = \{\mathcal{B}[1], ..., \mathcal{B}[T]\}$.
Consider any $A \in \mathcal{A}$, partition it into $A_1 + A_2$ as in Lemma \ref{lemimport} and partition $A_1$ into $A_F + A_B$ as in Corollary \ref{cor1}. Let $F_A[t] = \{ F_i[t] : e_i^\text{fwd} \in A_F \}$ denote the signals transmitted on $A_F$ at time $t$, and likewise let $B_A[t] = \{ B_j[t] : e_j^\text{bwd} \in A_B \}$. Let $a=|{A}_F|$, $b = |{A}_B|$, $c = |A_2|$. Recall that $f_{A_F}[t]$ are the signals sent by all backward edges to the edges in $A_F$ at time $t$, $M$ is the source message, and $K_D$ is all randomness generated by the sink. Now we upper bound the message rate $R_s$. It follows,
\begin{align}
  T R_s = H(M) & \stackrel{(a)}{=} H(M|K_D) - H(M|\mathcal{F}, \mathcal{B}, K_D)\nonumber\\
        & = I(M;\mathcal{F}, \mathcal{B}| K_D) \nonumber\\
        & = H(\mathcal{F}, \mathcal{B}| K_D) - H(\mathcal{F}, \mathcal{B} | M, K_D), \label{hh}
\end{align}
where (a) is due to the decoding constraint and the fact that $K_D$ is independent from $M$. We first study the first term in (\ref{hh}). Expand it according to the chain rule, we have
\begin{align}
  H(\mathcal{F}, \mathcal{B}| K_D) & = H(\mathcal{F}[1], ..., \mathcal{F}[T], \mathcal{B}[1], ..., \mathcal{B}[T] | K_D)\nonumber\\
  & \hspace{-20mm}\stackrel{(b)}{=} \sum_{i=1}^T H(\mathcal{F}[i],\mathcal{B}[i] | \mathcal{F}[0...i-1], \mathcal{B}[0...i-1],K_D)\nonumber\\
  & \hspace{-20mm}\stackrel{(c)}{=} \sum_{i=1}^T H(\mathcal{F}[i] | \mathcal{F}[0...i-1], \mathcal{B}[0...i-1],K_D)\nonumber\\
  & \hspace{-20mm}\stackrel{(d)}{\le} \sum_{i=1}^T H(\mathcal{F}[i] \backslash F_A[i] | \mathcal{F}[0...i-1], \mathcal{B}[0...i-1], K_D) \nonumber\\
  & \hspace{-9mm} + H(F_A[i]|\mathcal{F}[0...i-1], \mathcal{B}[0...i-1], K_D)\nonumber\\
  & \hspace{-20mm}\stackrel{(e)}{\le} T(x-a) +\sum_{i=1}^T H(F_A[i]|F_A[0...i-1],B_A[0...i-1]) \nonumber\\
  & \hspace{-20mm}\stackrel{(f)}{=} T(x-a) +\sum_{i=1}^T H(F_A[i]|F_A[0...i-1],B_A[0...i-1],M)\label{hfbkd}
\end{align}
Here (b) follows from the chain rule; (c) follows from the fact that $\mathcal{B}[i]$ is a function of the conditions; (d) follows from the chain rule and conditioning reduces entropy; (e) follows from conditioning reduces entropy; and (f) follows from the secrecy constraint, i.e., $M$ is independent from $F_A[0...T], B_A[0...T]$. Next we deal with the second term in (\ref{hh}).
\begin{align}
  H(\mathcal{F},\mathcal{B}|M,K_D) &  \ge H(F_A[1...T], B_A[1...T]| M, K_D)\nonumber\\
  &\hspace{-25mm} = \sum_{i=1}^T H(F_A[i],B_A[i]|F_A[0...i-1],B_A[0...i-1],M,K_D)  \nonumber\\
  &\hspace{-25mm} \ge \sum_{i=1}^T H(F_A[i]|F_A[0...i-1],B_A[0...i-1],M,K_D) \nonumber\\
  &\hspace{-25mm} \stackrel{(g)}{\ge} \sum_{i=1}^T H(F_A[i] | F_A[0...i-1], B_A[0...i-1],M) \label{hfbmkd}\\
  & \hspace{-20mm} - I(f_{A_F}[0...i-1];F_A[i]|F_A[0...i-1], B_A[0...i-1],M) \nonumber
\end{align}
Where (g) is due to Lemma \ref{takeoutD} by regarding $F_A[i]$ as $X$; $f_{A_F}[0,...,i-1]$ as $Y$; $K_D$ as $Z$; and $M, F_A[0,...,i-1], B_A[0,...,i-1]$ as $W$. Note that indeed $F_A[i]$ learns everything it can about $K_D$ from $f_{A_F}[0,...,i-1]$. Plug (\ref{hfbkd}) and (\ref{hfbmkd}) into (\ref{hh}) yields,
\begin{multline}
  T R_s \le T(x-a) \\+ \sum_{i=1}^{T-1} I(f_{A_F}[1...i];F_A[i+1]|F_A[1...i],B_A[1...i],M)  \label{rslong}
\end{multline}
Finally we bound the mutual information terms that appear in (\ref{rslong}). These terms characterize how the sink generated keys at times $1,...,i$ contribute to randomizing (and therefore protecting) the forward signals transmitted at time $i+1$.
\begin{align}
&\sum_{j=1}^{T-1} I(f_{A_F}[1...j];F_A[j+1]|F_A[1...j],B_A[1...j],M)\nonumber\\
   & \stackrel{(h)}{=}  \sum_{j=1}^{T-1} \sum_{i=1}^j I( f_{A_F}[i] ; F_A[j+1] | F_A[1...j],\nonumber\\& \hspace{40mm} B_A[1...j], f_{A_F}[0...i-1] ,M ) \nonumber\\
  & \stackrel{(i)}{=}  \sum_{i=1}^{T-1} \sum_{j=i}^{T-1} I( f_{A_F}[i] ; F_A[j+1] | F_A[1...j],\nonumber\\& \hspace{40mm} B_A[1...j], f_{A_F}[0...i-1] ,M ) \nonumber\\
  & \stackrel{(j)}{\le} \sum_{i=1}^{T-1} I(f_{A_F}[i] ; F_A[i+1] | F_A[1...i], B_A[1...i],f_{A_F}[0...i-1],M )\nonumber\\
  & \hspace{5mm} + \sum_{i=1}^{T-2} \sum_{j=i+1}^{T-1} I( f_{A_F}[i] ; F_A[j+1], B_A[j] \nonumber\\ & \hspace{20mm} | F_A[1...j], B_A[1...j-1], f_{A_F}[0...i-1] ,M ) \nonumber\\
  & \stackrel{(k)}{=} \sum_{i=1}^{T-1} I( f_{A_F}[i] ; F_A[i+1...T], B_A[i+1...T-1] \nonumber\\ & \hspace{31mm} | F_A[1...i], B_A[1...i], f_{A_F}[0...i-1] , M )  \nonumber\\
  & \stackrel{(l)}{\le} \sum_{i=1}^{T-1} H(f_{A_F}[i] | B_A[i])\nonumber\\
  & \stackrel{(m)}{\le} (T-1)(\text{rank}(\bar{U}_A) - b - c)\label{mutual}
\end{align}
Here (h) follows from the chain rule for mutual information; (i) follows from changing the order of summation; (j) follows from the fact that $I(X;Y|Z) \le I(X;Y,Z)$; (k) follows from the chain rule for mutual information; (l) follows from the definition of mutual information and conditioning reduces entropy; and (m) follows from Corollary \ref{cor1}. Finally substitute (\ref{mutual}) into (\ref{rslong}) we have
\begin{IEEEeqnarray*}{+rCl+x*}
  R_S  & \le & \frac{T(x - a + \text{rank}(\bar{U}_A) - b - c ) - \text{rank}(\bar{U}_A) + b + c}{T}\\
  & = & \frac{T(x + \text{rank}(\bar{U}_A) - |A|) - \text{rank}(\bar{U}_A) + b + c}{T}\\
  & < & x + \text{rank}(\bar{U}_A) - |A| & \QED
\end{IEEEeqnarray*}\renewcommand{\IEEEQED}{}
\end{proof}

\section{Achievability}
In this section we construct a scalar linear code that achieves the upper bound of Theorem \ref{th1} in $\bar{\mathcal{G}}$, thereby finding the secrecy capacity of $\bar{\mathcal{G}}$. The achievability result also implies that the upper bound is optimal if one only looks at the cut and its connectivity matrix. We will build the code on top of $\bar{C}$, with the idea that $\bar{C}$ is rank maximized and therefore suggests an ``optimal'' way of using the backward keys (i.e., sink generated randomness) to provide maximum randomization and protection. Hence what remains to be designed is the forward keys (source generated randomness that is independent from the message), and it turns out that $k_f =  \max_{A \in \mathcal{\mathcal{A}}} |A| - \text{rank}(\bar{U}_A)$ units of forward keys are sufficient in $\bar{\mathcal{G}}$. Therefore a rate of $R_s = x - k_f = x + \min_{A \in \mathcal{A}} \text{rank}(\bar{U}_A) - |A| $ can be achieved.

For the ease of presentation we start with the assumption that there is no delay in $\bar{\mathcal{G}}$, and will construct a code that achieves capacity exactly. We will show later that extending this code to networks with delay is straightforward, and in this case it achieves capacity asymptotically.

Let $m_1, ..., m_{R_s}$ be the messages, $K_S^1, ..., K_S^{k_f}$ be the source generated keys, $K_D^1, ..., K_D^{y}$ be the sink generated keys, all of them are i.i.d. uniformly distributed in $\mathbb{F}_q$. Let $E = (e_{ij}) \in \mathbb{F}_q^{(x+y) \times (x+y)}$ be the encoding matrix, defined by
\begin{equation}\label{Encoder}
  E = \left( \begin{array}{l|l}
    \begin{array}{l}
      G \\ \hline  \bm{0}
    \end{array}
     & \bar{C}
  \end{array}
   \right),
\end{equation}
where $G$ is a random matrix of size $x \times x$ with entries i.i.d. uniformly chosen from $\mathbb{F}_q$, and $\bm{0}$ is a zero matrix of size $y \times x$. Then the signals transmitted on the cut is
\begin{align*}
  \left( \begin{array}{l}
    F_1 \\ \vdots \\ F_x \\ B_1 \\ \vdots \\ B_y
  \end{array} \right) = E \left(
  \begin{array}{l}
    m_1\\ \vdots \\ m_{R_s} \\ K_S^1 \\ \vdots \\ K_S^{k_f} \\ K_D^1 \\ \vdots \\ K_D^{y}
  \end{array}
   \right)
\end{align*}
Notice that $E$ is a full rank square matrix with high probability since $G$ is generic and the bottom $y$ rows of $\bar{C}$ are linearly independent. Therefore the sink $D$ can decode everything. We only need to show the code is secure, i.e.,  any $z$-subset of $\{F_1, ..., F_x, B_1, ..., B_y\}$ is independent from $\{m_1, ..., m_{R_s}\}$. Since linearly independence implies independence, it suffices to show that the row space of
\begin{align*}
  E_{\text{Message}} = (I_{R_s} \ | \ \bm{0}_{R_s \times ( k_f + y)} )
\end{align*}
and the row space of any $E_A$, $A \in \mathcal{A}$ intersect trivially, where $E_A$ is the submatrix of $E$ formed by the rows that correspond to the edges in $A$. Let $E^r$ be the submatrix of $E$ by deleting the first $R_s$ columns from $E$, then it suffices to show any submatrix $E^r_A$, $A \in \mathcal{A}$ has full row rank. The following theorem shows this is true when $q$ is sufficiently large.

\begin{theorem}
  The code $E$ is secure with probability at least $1 - |\mathcal{A}| \frac{k_f (x+y)}{q}$.
\end{theorem}
\begin{proof}
  As mentioned above it suffices to show that any $E_A^r$, $A \in \mathcal{A}$ has full row rank. Consider an arbitrary $E_A^r$, and notice that the last $y$ columns of it is exactly $\bar{U}_A$.
  Assume that the $A$ contains $a$ forward edges and $b$ backward edges. Then due to the structure of $\bar{C}$, the last $b$ rows of $\bar{U}_A$ must be linearly independent. There exist $\text{rank}(\bar{U}_A) - b$ rows among the first $a$ rows such that these rows and the last $b$ rows together form a basis of the row space of $\bar{U}_A$.
  The remaining $|A|- \text{rank}(\bar{U}_A)$ rows of $\bar{U}_A$, all of them correspond to forward edges, are in the linear span of this basis. Assuming without loss of generality that they are the first $|A|- \text{rank}(\bar{U}_A)$ rows (otherwise reorder the forward edges), we construct a matrix $E^+_A$ as
  \begin{align*}
    E^+_A = \left( \begin{array}{l|l}
      \begin{array}{l}
        I_{|A|- \text{rank}(\bar{U}_A)} \\ \bm{0}
      \end{array}
       & \bar{U}_A
    \end{array} \right)
  \end{align*}
  Notice that $k_f = \max_{A' \in \mathcal{A}} |A'| - \text{rank}(\bar{U}_A') \ge |A| - \text{rank}(\bar{U}_A)$, so the number of columns in $E^+_A$ is not larger than the number of columns in $E^r_A$.
  We append $k_f +  \text{rank}(\bar{U}_A) - |A|$ more zero columns to the left of $E^+_A$ to obtain a matrix $E^{++}_A$ which is of the same size as $E^r_A$.
  Note that $E^{++}_A$ has full row rank and satisfies the zero block constraint as defined in (\ref{Encoder}). Hence $E^{++}_A$ contains a $|A| \times |A|$ full rank submatrix, denoted by $\underline{E}_A$. Consider the polynomial matrix $\underline{E}_A[\bm{x}]$ as a submartix of $E_A^r$ by regarding all random entries in $E_A^r$ as indeterminates, then it follows that $\det (\underline{E}_A[\bm{x}])$ is not the zero polynomial because $\det (\underline{E}_A) \ne 0$.  By the Schwartz-Zippel lemma, under the random selection of entries in $E^r_A$,
  \begin{align*}
    \Pr\left\{ \det(\underline{E}_A[\bm{x}]) = 0 \right\} \le \frac{k_f |A|}{q}
  \end{align*}

  Finally by the union bound, 
    \begin{IEEEeqnarray*}{+rCl+x*}
    \Pr\left\{ \bigcup_{A \in \mathcal{A}} \det(\underline{E}_A[\bm{x}]) \ne 0 \right\}  & \ge &  1 - \sum_{A \in \mathcal{A}} \Pr\{ \det(\underline{E}_A[\bm{x}]) = 0 \}\\
    & \ge & 1 - |\mathcal{A}| \frac{k_f (x+y)}{q} & \hspace{-8mm} \QED
  \end{IEEEeqnarray*}\renewcommand{\IEEEQED}{}
\end{proof}

Extending the above code to networks with delay is straightforward. It suffices for the source to wait one time slot for the arrival of the first batch of keys, and then start transmitting normally. So the overhead is vanishing as we increase the time duration of the code.

\section{Conclusion}
We consider the problem of secure communication over a network in the presence of wiretappers. We gives a cut-set bound of secrecy capacity which takes into account the network connectivity and the contribution of backward edges. We show the bound is tight on a class of networks. One interesting problem that future works may study is to improve the cut-set bound with more network characteristics, such as the min cut from backward edges to forward edges, which may quantify any bottlenecks in the use of backward edges.





%
\bibliographystyle{IEEEtran}
\bibliography{bbndbib}

\end{document}